\documentclass[11pt]{article}
\usepackage[margin=1in]{geometry}

\usepackage[utf8]{inputenc}
\usepackage[english]{babel}
\usepackage{lipsum}
\usepackage{enumitem}
\usepackage{bm}
\usepackage{comment}
\usepackage{upgreek}
\usepackage{bbold}
\usepackage{array}
\usepackage{amsmath}
\usepackage{graphicx}
\usepackage{algcompatible}
\usepackage{amssymb}
\usepackage{mathtools}
\usepackage{algorithm}
\usepackage{amsthm}
\usepackage{relsize}
\usepackage{standalone}
\usepackage{tikz}
\usetikzlibrary{positioning}

\newtheorem{theorem}{Theorem}[section]
\newtheorem{lemma}[theorem]{Lemma}

\newtheorem{definition}[theorem]{Definition}

\usepackage{microtype} %improves the spacing between words and letters

\usepackage[svgnames]{xcolor} % Enabling mixing colors and color's call by 'svgnames'

\definecolor{MyColor1}{rgb}{0.2,0.4,0.6} %mix personal color

\newcommand{\yaron}[1]{}
\newcommand{\seffi}[1]{}

\newcommand{\etal}{et al.~}

\usepackage{titlesec}
\usepackage{caption}
\usepackage{subcaption}
\newcommand{\set}[1]{\left\{ #1 \right\}}

\usepackage{xspace}
\newcommand{\OPT}{\textit{OPT}\xspace}

\newcommand{\pwtw}{paging with time windows\xspace}
\newcommand{\cost}{cost}

\makeatletter
\let\reftagform@=\tagform@
\def\tagform@#1{\maketag@@@{(\ignorespaces\textcolor{red}{#1}\unskip\@@italiccorr)}}
\renewcommand{\eqref}[1]{\textup{\reftagform@{\ref{#1}}}}
\makeatother
\usepackage{hyperref}
\hypersetup{colorlinks=true}

\usepackage{fancyhdr}

\usepackage{multirow}

\title{Tight Bounds for Paging with Time Windows
and Joint Replenishment with Varying Opening Costs}
\author{
Joseph (Seffi) Naor%
\thanks{Supported in part by ISF grant 3001/24
and United States--Israel BSF grant 2022418.}\\
Computer Science Department, Technion\\
Haifa, Israel\\
\texttt{naor@cs.technion.ac.il}
\and
Tomer Tsachor\\
Computer Science Department, Technion\\
Haifa, Israel\\
\texttt{tomer.ts@cs.technion.ac.il}
}
\date{\today}

\begin{document}

\begin{titlepage}
\clearpage\thispagestyle{empty}

\maketitle
\thispagestyle{empty}

\begin{abstract}
We study the paging with time windows problem, a natural generalization of classical paging, in which each request must be served within a specified time interval. This model allows the algorithm to delay service and batch service requests over time, but also introduces significant challenges due to overlapping requests. While the problem is well understood in the unweighted setting, where optimal $O(\log k)$-competitive algorithms are known, the weighted case remains more elusive. The best previously known algorithm, due to Gupta, Kumar, and Panigrahi~\cite{GKP20}, achieves an $O(\log k \log n)$-competitive ratio, leaving a gap between weighted and unweighted settings.

We obtain an $O(\log k)$-competitive algorithm for weighted paging with time windows, matching the optimal competitive ratio. The core of our approach is a structural decomposition that separates the evolution of the cache from the scheduling of page loads. We show that, given a sequence of cache configurations, computing a near-optimal sequence of page loads reduces to a variant of the Joint Replenishment Problem with deadlines (JRP-D) with time-dependent opening costs. We design a constant-competitive online algorithm for this variant and prove that the reduction is tight. Combining this with an existing online construction of cache configurations yields the desired bound.
\end{abstract}
\end{titlepage}

\section{Introduction}

In the {\em classical weighted paging problem} there is a universe of $n$ pages, where each page \(p\) has a weight \(w_p\), and there is a cache that can hold up to \(k\) pages. At each time step, a page is requested, and if the requested page
is already in the cache then no cost is incurred, otherwise,
the page must be loaded into the cache, incurring a cost equal to its weight. The goal is to minimize the total cost incurred. 
In the unweighted version of the problem, all weights are equal (unit weights).
%We refer to this problem as {\em classical}.

Paging is one of the earliest and most extensively studied problems 
in online computation and competitive analysis ~\cite{ST85,F+,Y91,Y98,GS,BBN07,BBN08,ACN,AAK99,BBK,BFT96,I98,I97,I02}, 
including works on non-standard caching models, e.g., elastic caches~\cite{GuptaK0P19}, caching with time windows~\cite{GKP20}, caching
with dynamic weights~\cite{EvenMR18}, and caching with machine learning predictions~\cite{LykourisV18}. In fact, online paging has become a focal point for many of the recent developments in competitive analysis, e.g., the
online primal-dual method, projections, and mirror descent~\cite{BNsurvey, BCN14, BCLLM18}. 

In their seminal paper, Sleator and Tarjan \cite{ST85}
showed that any deterministic online algorithm is at least $k$-competitive and 
that the LRU policy (Least Recently Used) is exactly $k$-competitive for unweighted paging. 
The $k$-competitive bound was later generalized to weighted paging as well~\cite{CL91,Y94}.
When randomization is allowed, Fiat \etal~\cite{F+} gave the elegant randomized marking algorithm for unweighted paging, which is $\Theta(\log k)$-competitive against an oblivious adversary. Bansal \etal~\cite{BBN07} gave a $\Theta(\log k)$-competitive randomized algorithm 
for weighted paging based on the online primal-dual framework~\cite{BNsurvey}. 

We now elaborate on the two problems on which we focus in the paper.

\vspace{-0.1in}

\paragraph*{Paging with time windows.}
A natural and well-studied generalization is \emph{paging with time windows}, where each request must be served within a specified interval. Formally, a request $r$ is represented by a triple $r=(p(r),s(r),e(r))$, where $p(r)\in P$ is the requested page, $s(r)$ is the release time, and $e(r)\ge s(r)$ is the deadline for serving it. A request $r$ is satisfied if page $p(r)$ is present in the cache at some time during the interval $[s(r),e(r)]$. 

This model significantly increases the flexibility of any paging algorithm, allowing it to delay service, while batching together the service of several requests with overlapping time windows. At the same time, it introduces new challenges, as requests may overlap and interact in complex ways.

In the unweighted setting, Azar et al.~\cite{azar} showed that paging with time windows admits an \(O(k)\)-competitive deterministic algorithm and an \(O(\log k)\)-competitive randomized algorithm. The latter matches the optimal randomized guarantee for classical paging up to constant factors. At a high level, the randomized algorithm delays loading a page until the deadline of one of its pending requests, thereby batching together overlapping requests for the same page.

However, in the weighted setting, the authors of~\cite{azar} showed that this structural approach no longer holds, although the \(\Omega(k)\) lower bound for deterministic algorithms continues to hold. Gupta et al.~\cite{GKP20} subsequently introduced a novel LP formulation based on sophisticated covering constraints. They showed that paging with time windows is NP-hard and APX-hard even when $k=1$ and all pages have unit weight. Nevertheless, their LP has a constant integrality gap, yielding an $O(1)$-approximation algorithm for the problem. In the online setting, they obtained an \(O(\log n \log k)\)-competitive algorithm for weighted paging with time windows, which remains the best known bound to date.

Gupta et al.~\cite{GKP20} also considered the more general problem of \emph{paging with delays}. Here, each request $r$ is represented by a triple
\(
r=(p(r),s(r),F_r),
\)
where $p(r)$ is the requested page, $s(r)$ is the release time, and \(
F_r:\{s(r),s(r)+1,\ldots\}\to\mathbb{R}_{\ge0}
\) is a non-decreasing delay-loss function associated with $r$. If request $r$ is served at time $t\ge s(r)$, the algorithm incurs delay cost $F_r(t)$. The objective is to minimize the total weight of the loaded pages plus the total delay costs incurred by all requests.

This formulation strictly generalizes paging with time windows: a request with deadline \(d\) and release time $s$ corresponds to the delay function that is \(0\) throughout the interval \([s,d]\) and \(\infty\) afterwards. Gupta et al.~\cite{GKP20} showed how to extend their algorithmic framework from paging with time windows to paging with delays, while preserving the same asymptotic competitive guarantees. 

We focus primarily on the time-window formulation in the paper in order to present the main ideas more cleanly, though our techniques naturally extend to the more general setting of paging with delays.

\paragraph*{Joint Replenishment Problem.}

A well-studied problem that turns out to be closely related to paging with time windows is the Joint Replenishment Problem with deadlines (JRP-D). The problem can be viewed as a depth-$2$ tree (see Figure~\ref{fig:jrp-tree}), consisting of a warehouse at the root and a single leaf for each item, with a bijection between items and leaves. Requests for items arrive over time and an item can be requested multiple times.  At any time $t$, an order can be \emph{placed} by opening the warehouse and selecting a subset of the leaves with pending requests for items. A request for an item is \emph{served} at time $t$ if its corresponding leaf is included in the order. Opening the warehouse incurs a fixed cost, while serving requests for item $i$ incurs an additional item-specific cost $w_i$.

Requests for items must be served within prescribed time intervals. Equivalently, serving a request after its deadline incurs infinite cost. An algorithm may place an order at any time and use it to serve multiple pending requests simultaneously. Thus, the main trade-off is between serving requests immediately vs. delaying service to batch multiple requests together.  While the opening cost may be smaller than the item costs in general, the most interesting regime is when the opening cost is large, making it advantageous to postpone service and aggregate requests belonging to different items into a single order. For this problem, Bienkowski \etal~\cite{bienkowski2014better} gave a $2$-competitive online algorithm, matching the optimal competitive ratio. 

JRP-D is a special case of the more general Joint Replenishment Problem (JRP), in which requests incur arbitrary delay costs while awaiting service, rather than hard deadlines. JRP and its variants have been extensively studied in both offline~\cite{arkin1989computational,becchetti2009latency} and online settings. In particular, Buchbinder \etal~\cite{buchbinder2013online} showed that JRP admits a $3$-competitive online algorithm.

In this work, we study a variant of JRP-D with time-dependent warehouse opening costs. Unlike the classical setting, the algorithm must account not only for the possibility of batching requests, but also for changes in the cost of placing an order. Although time-dependent opening costs have previously been investigated experimentally~\cite{federgruen1994joint}, we examine this setting from the perspective of online competitive analysis. As we show, \pwtw admits a natural reduction to this JRP-D variant.

\begin{figure}[t]
\centering
\begin{tikzpicture}[scale=1, every node/.style={font=\small}]

\node[draw, circle] (r) at (0,1.8) {$r$};
\node[right=2pt of r] {$12$};

\node[draw, circle] (1) at (-2.4,0) {$1$};
\node[draw, circle] (2) at (-0.8,0) {$2$};
\node[draw, circle] (3) at (0.8,0) {$3$};
\node[draw, circle] (4) at (2.4,0) {$4$};

\node[right=2pt of 1] {$7$};
\node[right=2pt of 2] {$6$};
\node[right=2pt of 3] {$5$};
\node[right=2pt of 4] {$1$};

\draw (r) -- (1);
\draw (r) -- (2);
\draw (r) -- (3);
\draw (r) -- (4);

\end{tikzpicture}
\caption{A depth-$2$ tree view of JRP. The root represents the warehouse, with opening cost $12$, and each leaf represents an item together with its item cost. An order corresponds to selecting the root together with any subset of the leaves. For example, serving items $\{2,4\}$ incurs total cost $12+6+1=19$, while serving only item $4$ costs $12+1=13$.}
\label{fig:jrp-tree}

\end{figure}
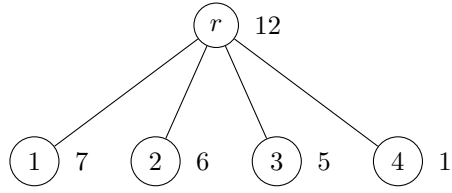
Throughout the paper, when discussing randomized algorithms, all costs and competitive guarantees are in expectation, unless stated otherwise.

\subsection{Our contribution}

Our first contribution is a new formalism for representing solutions to \pwtw. We represent a solution as a pair $(C,D)$, where $C=(C_t)$ denotes the sequence of cache configurations at the end of each time step $t$, and $D=(D_t)$ denotes the sequence of page loads performed during each time step $t$. This decomposition cleanly separates the evolution of the cache over time from the scheduling of page loads used to satisfy requests that are not served by the cache configurations in $C$. In particular, an algorithm may load a page during a time step solely to satisfy pending requests and then immediately evict it. Consequently, a page may belong to $D_t$ without belonging to $C_t$. Such page loads contribute to the cost of the solution, but are not reflected in the cache configurations.

This representation is purely conceptual: it does not alter the objective function, but rather decomposes a paging solution into two components, corresponding to cache evolution and page-load scheduling. Given a request sequence $R$, let $\cost_R(C,D)$ denote the cost of the solution represented by a feasible pair $(C,D)$. We define \(cost_R(C)=\min_D \cost_R(C,D)\), where the minimum is taken over all completions $D$ of $C$ that are feasible for $R$. When the request sequence is clear from context, we omit the subscript $R$ and write $\cost(C,D)$ and $\cost(C)$.

Furthermore, we show that the online cache-configuration construction implicit in~\cite{GKP20} yields a cache trajectory $C$, i.e., a sequence of cache configurations, satisfying

\begin{equation}\label{eq:good-cache-cost}
\cost(C)=O(\log k)\cdot\OPT,    
\end{equation}
where $\OPT$ denotes the cost of an optimal offline solution to \pwtw.

Our second contribution is to isolate the \emph{completion problem} as a problem in its own right. In the completion problem, the sequence of cache configurations $C$ is fixed, but revealed online, and the goal is to compute a feasible minimum-cost completion of $C$. We show that the completion problem can be formulated as an instance of JRP with deadlines (JRP-D) with time-dependent opening costs, i.e., opening costs can vary over time.

The varying opening costs make this problem fundamentally harder than classical JRP-D, as without additional information about the future no bounded competitive ratio is possible. To get around this difficulty, we introduce a prediction-augmented model in which the algorithm receives limited information regarding future opening costs. Specifically, for each request that has been released but not yet served, which we call an \emph{active request}, the prediction identifies the cheapest opening opportunity before its deadline. In fact, we show that an even weaker prediction suffices to obtain a constant-competitive online algorithm.

\begin{theorem}
There exists a constant-competitive online algorithm for JRP-D with time-dependent opening costs under our prediction model.
\end{theorem}

Our next contribution is to show that the prediction required by the JRP-D algorithm is naturally available in the reduction we provide from paging with time windows. 
In fact, the prediction used by the JRP-D algorithm in this case is simply the weight of the page whose pending request has the earliest deadline. Consequently, given any sequence of cache configurations $C$ revealed online, our algorithm computes online a completion $D$ satisfying 
\[
\cost(C,D)=O(\cost(C)).
\]
Combining this completion procedure with the cache-configuration bound in~\eqref{eq:good-cache-cost} immediately yields
\[
\cost(C,D)=O(\log k)\cdot\OPT,
\]
and therefore an $O(\log k)$-competitive online algorithm for \pwtw, improving upon the previous $O(\log k\log n)$-competitive bound of Gupta et al.~\cite{GKP20}.

\begin{theorem}
There exists a randomized $O(\log k)$-competitive online algorithm for weighted paging with time windows.
\end{theorem}

The same framework also yields an $O(\log k)$-competitive algorithm for the more general paging-with-delays model of Gupta et al.~\cite{GKP20}.

\begin{theorem}
There exists a randomized $O(\log k)$-competitive online algorithm for weighted paging with delays.
\end{theorem}

\subsection{Technical overview}

\paragraph{Separating cache configurations from page loads.}

A key contribution of Gupta et al.~\cite{GKP20} is a novel LP formulation for paging with time windows based on sophisticated covering constraints (see Section~\ref{subsection:prior-work} for details). Their framework proceeds in two stages. First, they compute an online fractional solution to their LP and round it to an integral solution. This rounding incurs only an $O(1)$ loss in the offline setting, but an $O(\log k)$ loss online. The resulting integral solution, however, is not necessarily a feasible solution. In the second stage, they convert it into a feasible paging solution. This conversion incurs only an additional $O(1)$ loss offline, yet an additional $O(\log n)$ loss online, resulting in an $O(\log k \log n)$-competitive algorithm.

Although this distinction is not made explicit in~\cite{GKP20}, a closer examination of ~\cite{GKP20} reveals that both online and offline conversion algorithms produce the exact  same cache configuration at the end of every time step. The only difference lies in how they schedule the ``extra" page loads used to satisfy pending requests. This raises a natural question: if both algorithms maintain the same cache configuration at the end of every time step, how does the additional $O(\log n)$ factor arise?

The answer is that in contrast to classical paging, the cache configurations alone do not fully determine the paging solution. In classical paging, one may assume without loss of generality that the only page loaded at each time step is the requested page, which then remains in cache. Consequently, if $C_t$ denotes the cache configuration at time $t$, then the paging cost incurred at time $t$ is exactly $W(C_t\setminus C_{t-1})$, where $W(S)$ denotes the total weight of the pages in a set $S$. Thus, in the classical setting, the sequence of cache configurations completely determines the solution cost.

In \pwtw, however, this is no longer the case. Since requests may be delayed and then batched together, an algorithm may load several pages simultaneously at a single time step. Moreover, a page may be loaded only to satisfy a pending request and then evicted immediately, without ever appearing in the cache configurations. As a consequence, the quantity $W(C_t\setminus C_{t-1})$ may be strictly smaller than the actual cost incurred at time $t$, and the sequence of cache configurations by itself no longer captures all load operations performed by an algorithm.

This distinction motivates our $(C,D)$ formalism. The cache trajectory records only the pages present at the end of each time step, whereas the completion records all page loads used to serve the request sequence, including (i) pages that are loaded and immediately evicted; and (ii) page loads that are implied by the changes in the cache configurations. Recall that $\cost(C,D)$ denotes the cost of the solution represented by $(C,D)$, and that $\cost(C)$ is the minimum cost over all completions of $C$. Thus, two solutions may induce the same cache trajectory $C$ while using different completions and incurring different costs.

Although Gupta et al.~\cite{GKP20} do not explicitly introduce this decomposition, it is implicit in their analysis. Through the lens of our formalism, their online conversion algorithm takes an integral solution and constructs a pair $(C,D)$ whose cost is within an $O(\log n)$ factor of that integral solution. At the same time, the same cache trajectory $C$ does admit (offline) a completion $D'$ such that $\cost(C,D')$ is within an $O(1)$ factor of the same integral solution. Thus, the additional $O(\log n)$ factor arises entirely from the online construction of $D$, rather than the cache trajectory $C$ itself. We formalize this reinterpretation in
Section~\ref{subsection:Reinterpretation}.

This perspective suggests focusing exclusively on the completion problem. Our goal is therefore the following.
Given a sequence of cache configurations $C$ revealed online, compute online a completion $D$ such that $\cost(C,D)=O(\cost(C))$.

\paragraph{A replenishment view of the completion problem.}

To compute such a completion $D$, we reinterpret the problem as a covering problem over time. Consider a request $r$ such that $p(r)\notin C_{s(r)}$. In order for a pair $(C,D)$ to be feasible, page $p(r)$ must appear in one of the sets $D_t$ for $t\in[s(r),e(r)]$. Thus, every such request naturally induces a covering constraint over its time window.

The second key observation concerns the cost of loading a set of pages while following a fixed cache trajectory $C$. Assume throughout that the cache is full at every time step, i.e., $|C_t|=k$ for all $t$. Suppose that at time $t$ we wish to load a set $S$ of pages in order to serve active requests, where $C_t\setminus C_{t-1}\subseteq S$.

First, suppose that $C_t\neq C_{t-1}$. The transition from $C_{t-1}$ to $C_t$ already creates space for loading pages: the pages in $C_{t-1}\setminus C_t$ are evicted, while the pages in $C_t\setminus C_{t-1}$ are loaded as part of $S$. We may therefore load all pages in $S$ while realizing the prescribed cache transition, with $D_t=S$.

In contrast, suppose that $C_t=C_{t-1}$. Since the cache is full, loading a page outside the cache requires temporarily evicting a page from $C_t$. We may evict a minimum-weight page $q\in C_t$, load and immediately evict the pages in $S$ (one by one), and finally reload $q$ so that the cache configuration at the end of the time step is again $C_t$. In this case, $D_t=S\cup\{q\}$, and the additional cost of performing the load is $w_q$.

This suggests viewing the completion problem as a replenishment problem, where loading the pages in $S$ corresponds to serving a set of items. If the cache configuration changes at time $t$, the opening cost is effectively $0$, since the pages in $S$ can be loaded as part of the mandatory cache transition. If instead the cache configuration remains unchanged, performing a load may require an additional cost of \(\min\set{w_p\mid p\in C_t}\), corresponding to evicting and subsequently reloading a minimum-weight page from $C_t$. This page plays the role of the warehouse in a JRP instance. 
For simplicity, we define the opening cost at every time $t$ as
\[
o_t:=\min\set{w_p\mid p\in C_t}.
\]
The JRP formulation may therefore charge slightly more than the corresponding paging load: at times when the cache changes, the opening cost is charged separately even though the cache transition itself already provides the required loading opportunity. However, this additional cost can be charged to the pages loaded as part of the cache transition, increasing the total cost by at most a factor of $2$.

Taken together, these observations naturally define an instance $J(C)$ of JRP-D with time-dependent opening costs. Each page corresponds to an item, represented by a leaf in the JRP tree (see Figure~\ref{fig:jrp-tree}). The covering constraints induced by the requests specify the time intervals during which the requests must be served. At every time $t$, an order specifies a set of pages to load. The cost of placing an order consists of the opening cost $o_t$ together with the costs of the items served by the order.

We show that this reduction faithfully captures the completion problem. Every completion of $C$ induces a feasible solution to $J(C)$ with comparable cost, and conversely, every solution to $J(C)$ can be transformed online into a feasible completion of $C$ with no asymptotic loss. Consequently, solving $J(C)$ is, up to a constant factor, equivalent to solving the completion problem.

\paragraph{JRP-D with time-dependent opening costs.}

The next challenge is therefore to design an online algorithm for JRP-D with time-dependent opening costs. Unlike the classical JRP-D setting, where the opening cost is fixed, in our model the opening costs may vary significantly over time. Without additional information, this makes the problem substantially harder: an online algorithm may place an expensive order shortly before a much cheaper opening opportunity occurs, leading to an unbounded competitive ratio.

To get around this difficulty, we consider a prediction-augmented setting that provides information about cheaper future opening opportunities before the earliest pending deadline, allowing an algorithm to decide whether to serve a request immediately or postpone service. As in paging with time windows, release and deadline times of each JRP-D request $r$ are denoted by $s(r)$ and $e(r)$, respectively. For times $t<t'$, define $O_{t,t'}:=\min_{\tau\in(t,t']}\set{o_\tau}$. For a request $r$ for item $i$, define $P(r,t):=\max\set{O_{t,e(r)},w_i}$. Thus, $P(r,t)$ may overestimate the minimum opening cost $O_{t,e(r)}$ when $w_i>O_{t,e(r)}$.

The prediction value of a request $r$ at time $t$, denoted by $\widetilde P(r,t)$, is defined as follows. If there exists a request $r'$ whose time window satisfies $[s(r'),e(r')]\subseteq(t,e(r)]$, then $\widetilde P(r,t)$ may take any value in $[0,P(r,t)]$; otherwise, $\widetilde P(r,t)=P(r,t)$. The algorithm has access to $\widetilde P(r,t)$ only when $r$ has the earliest deadline among the requests that are active at that point. Thus, as requests are served, predictions for later-deadline active requests become available. Using these predictions, the algorithm decides whether a current earliest-deadline request should be served immediately or deferred in anticipation of a cheaper future opening opportunity.

Importantly, the required predictions are naturally available in the paging reduction. Consider an active page request $r$ for page $p$. We show that, if no request that has not yet arrived has a deadline earlier than $e(r)$, then one may assume\footnote{For intuition, suppose $p$ is the only page with deadline $e(r)$. Any feasible solution may postpone loading $p$ until time $e(r)$. Furthermore, if $p$ is evicted at time $e(r)$ in order to load another page $q$, then the same load can instead be performed at time $e(r)+1$. This preserves feasibility, since no other request has deadline $e(r)$, and incurs the same cost. Hence, there exists an optimal solution in which $p$ belongs to the cache at time $e(r)$. The general case, without the uniqueness assumption, is addressed in Lemma~\ref{lem:deadline-cache-repair}.} that $p\in C_{e(r)}$. Therefore, $O_{t,e(r)}\le o_{e(r)}=\min_{q\in C_{e(r)}}w_q\le w_p$, and hence $P(r,t)=\max\set{O_{t,e(r)},w_p}=w_p$. Thus, whenever the prediction is required to be exact, the paging reduction naturally provides it: the prediction value is simply the weight of the requested page.

When a request that has not yet arrived has an earlier deadline, the exact value $P(r,t)$ may depend on information that is not yet available at $t$. In this case, our prediction model permits an arbitrary underestimation of $P(r,t)$. In particular, the reduction may provide the value $w_p$, which is valid since $w_p\le P(r,t)$. Thus, even when the exact prediction value cannot yet be determined, the reduction provides a valid prediction to the JRP algorithm.

Using this prediction model, we design a constant-competitive online algorithm for JRP-D with time-dependent opening costs. Combining the reduction with our JRP algorithm yields an online procedure that, given a sequence of cache configurations $C$, computes a completion $D$ satisfying $\cost(C,D)=O(\cost(C))$.

Combining this completion procedure with the cache-configuration bound in~\eqref{eq:good-cache-cost} yields
\[
\cost(C,D)=O(\log k)\cdot\OPT,
\]
and therefore an $O(\log k)$-competitive online algorithm for paging with time windows.

Our techniques also extend to the more general \emph{paging-with-delays} model introduced by Gupta et al.~\cite{GKP20}. They first show that paging with delays is equivalent to a variant of paging with time windows in which requests can remain unserved at the cost of paying a request-specific penalty. Then they extend their analysis of paging with time windows to this penalty model. We show that our analysis admits the same extension, yielding an $O(\log k)$-competitive algorithm for paging with a general delay function as well.

\paragraph*{Organization.}
The remainder of the paper is organized as follows. In Section~\ref{sec:preliminaries}, we define the problems we consider, introduce our approach, and expand on previous work. In Section~\ref{sec:completion-reduction}, we define the reduction from \pwtw to JRP-D with time-dependent opening costs. In Section~\ref{sec:jrp-algorithm}, we present and analyze our constant-competitive algorithm for JRP-D with time-dependent opening costs. Finally, Appendix~\ref{app:delays} extends our framework to paging with delays.

\section{Our Model and Approach}\label{sec:preliminaries}
\subsection{Paging with Time Windows}\label{subsec:new-represntation}

In the \pwtw problem, there is a universe of $n$ pages, denoted by $P=\{p_1,p_2,\dots,p_n\}$. Each page $p\in P$ is associated with a positive weight $w_p$, and there is a cache of size $k$. The system evolves over discrete time steps $t=1,2,\dots,T$. At each time $t$, the cache stores a subset of pages $C_t\subseteq P$, referred to as the \emph{cache configuration} at time $t$. Without loss of generality, we assume that $|C_t|=k$ at every time $t$, since there is no benefit to evicting a page unless another page is loaded in its place.

\paragraph{Requests.}
Requests arrive over time. Each request $r = (p(r), s(r), e(r))$
%\[
%r = (p(r), s(r), e(r))
%\]
specifies a page $p(r)$, a release time $s(r)$, and a deadline $e(r)$, defining a time window $[s(r), e(r)]$. In the online setting, a request $r$ becomes known only at time $s(r)$. A request $r$ is served if the page $p(r)$ is present in the cache at some time $t \in [s(r), e(r)]$. Note that pages can be loaded and evicted at the same time step, thus allowing for multiple requests to be served simultaneously.

\paragraph{Cost.}
The cost of an algorithm is the total cost of all page loads over the time horizon. Loading a page $p$ into the cache incurs cost $w_p$, while evictions are free. Some prior work instead charges for page evictions; the two cost measures are equivalent up to an additive term determined by the initial and final cache configurations. The objective is to minimize the total cost while serving all requests.

The classical paging problem corresponds to the special case in which every request satisfies $s(r) = e(r)$.

\subsection{Prior Work: Gupta et al.~\cite{GKP20}}\label{subsection:prior-work}

The main difficulty in extending classical paging to the time-window setting arises from limitations of standard linear programming (LP) relaxations. In classical weighted paging, a request for page $p$ at time $t$ requires the page to reside in cache at time $t$. In contrast, in paging with time windows, a request $r$ only requires that $p(r)$ appear in the cache at some time during the interval $[s(r), e(r)]$.

A natural LP relaxation would therefore require that the total fractional residency of page $p(r)$ over the interval $[s(r), e(r)]$ be at least $1$. However, Gupta et al.~\cite{GKP20} showed that this formulation suffers from an integrality gap of $\Omega(k)$. Intuitively, such a fractional solution can satisfy a request by keeping a small fraction of the page in cache throughout the entire window, while an integral solution must eventually pay the full cost $w_p$.

To overcome this issue, Gupta et al.~\cite{GKP20} introduced a fundamentally different LP formulation based on sophisticated covering constraints, which they call the hitting-set LP. For each page $p$ and time $t$, the variable $x_{p,t} \in [0,1]$ represents the extent to which page $p$ is evicted at time $t$.

At a high level, the LP guarantees cache capacity through covering constraints over time intervals. Informally, if an interval $I$ contains requests to $k+1$ distinct pages, then the cache cannot simultaneously retain all of them throughout $I$, and therefore at least one of these pages must be evicted during the interval. The LP utilizes this observation by requiring that every such pair $(I,S)$, where $S$ is a set of $k+1$ pages requested during $I$, be ``hit'' by an eviction within the time interval of one of the pages in $S$.

In the standard LP formulations for paging, the LP directly specifies which pages reside in cache at each time step. In contrast, the hitting-set LP only determines when pages are evicted, and does not explicitly specify how requests are served. Consequently, an integral solution to the hitting-set LP does not immediately define a feasible paging schedule. Solving \pwtw therefore requires addressing two separate challenges: first, rounding the fractional LP solution online, and second, transforming the resulting eviction schedule into a feasible paging solution.

Gupta et al.~\cite{GKP20} showed that the hitting-set LP admits an online rounding procedure with an \(O(\log k)\) loss and an offline rounding procedure with only an \(O(1)\) loss, both relative to the optimal fractional solution. Their approach exploits the interval structure of the constraints and refines the primal-dual techniques developed by Bansal et al.~\cite{BNsurvey} for the classical weighted paging LP.

To bridge the gap between the rounded eviction schedule and a feasible paging solution, Gupta et al.~\cite{GKP20} introduced an additional scheduling layer. At each time step $t$, the algorithm determines whether a page must be loaded into the cache. Before updating the cache configuration, the algorithm may also perform additional ``load-and-evict'' operations: pages are loaded in order to satisfy requests and are evicted within the same time step, thus not appearing in \(C_t\).

Implementing this scheduling layer is precisely what separates the online and offline guarantees. The offline algorithm uses future knowledge of the cache trajectory to coordinate these ``load-and-evict'' operations more efficiently, achieving only an $O(1)$ loss relative to the rounded hitting-set LP solution. In contrast, the online algorithm incurs an additional $O(\log n)$ factor, resulting in an overall \(O(\log k \log n)\)-competitive ratio.

This gap is the starting point of our work. It suggests that the remaining difficulty in obtaining an optimal online algorithm lies not in constructing the cache trajectory itself, but in scheduling the additional page loads needed to satisfy requests.

\subsection{Our Approach}\label{subsec:our-approach}

The discussion above suggests that paging with time windows can be naturally decomposed into two tasks: determining how the cache evolves over time, and then deciding which additional pages to load in order to serve unsatisfied requests. While these two tasks are intertwined in the online algorithm of Gupta et al.~\cite{GKP20}, their offline algorithm suggests that they can be partially decoupled: after the cache configurations are determined, a final scheduling step is used to complete the solution by specifying the additional page loads.

Our approach is to formalize this separation. Intuitively, the cache trajectory captures the structural decisions of the algorithm, while the loading schedule determines how requests are served. The cache trajectory is not known in advance, but is revealed online one configuration at a time. Given this trajectory, the remaining task is to construct a loading schedule online that satisfies all requests at minimum cost.

At first glance, this is an interval-covering problem: each request induces a time interval during which its page must be loaded, and loading a page can simultaneously cover multiple overlapping intervals. However, the cache trajectory makes loading pages at certain times inherently more expensive than at other times, as a load may require an additional page to be evicted and then reloaded in order to preserve the given cache configuration. This creates a batching trade-off analogous to that in JRP-D, where serving a set of items requires paying an opening cost (of the warehouse) in addition to the item costs.

\paragraph{Solutions.}
We represent a solution using two components. Let $C=(C_t)_{t=0}^T$ denote a sequence of cache configurations, where $C_t\subseteq P$ is the set of pages in the cache at the end of time $t$. Let $D=(D_t)_{t=1}^T$ denote a sequence of page loads, where $D_t\subseteq P$ is the set of pages loaded at time $t$. A page belongs to $D_t$ whenever it is loaded at time $t$, i.e., if it is reloaded at the same time step that it is evicted, or if it is immediately evicted after it is loaded.

\paragraph{Feasibility.}
A pair $(C,D)$ is a feasible solution if it satisfies the following conditions:

\begin{itemize}
    \item \emph{Capacity.} For every time $t$, $|C_t| = k$.
    
    \item \emph{Consistency.} For every time $t$, $C_t \setminus C_{t-1} \subseteq D_t$. This implies that, for any page $p$, if there exist times $t_1 < t_2$ such that $p \notin C_{t_1}$ but $p \in C_{t_2}$, then there exists a time $t' \in (t_1, t_2]$ such that $p \in D_{t'}$.

  \item \emph{Non-degeneracy.} For every time $t$, if $D_t \neq \emptyset$, then $D_t \cap C_t \neq \emptyset$. That is, whenever pages are loaded at time $t$, at least one of the loaded pages remains in cache at the end of the time step. For example, if $C_t=C_{t-1}$ and a page $p\in C_t$ is temporarily evicted to create space and subsequently reloaded, then $p\in D_t$.

    \item \emph{Service.} For every request $r$, the requested page $p(r)$ must appear in the cache at some time during the interval $[s(r), e(r)]$. Equivalently, a request $r$ is served if either $p(r)$ is already present in the cache at the beginning of its window, or it is loaded at some time during the window. Formally, $p(r) \in 
C_{s(r)}
\;\cup
\left(\bigcup_{t=s(r)}^{e(r)} D_t\right)
$.
\end{itemize}

\paragraph{Cost.}
The cost of a feasible pair $(C,D)$ with respect to a request sequence $R$ is
\(
\cost_R(C,D)=\sum_{t=1}^T W(D_t).
\) For an infeasible pair $(C,D)$, we define $\cost_R(C,D)=\infty$.

Observe that by the consistency condition, every page that is not in $C_{t-1}$, but belongs to $C_t$, must be loaded at time $t$ and hence belongs to $D_t$. Consequently, any feasible pair $(C,D)$ satisfies
\(
W(C_t\setminus C_{t-1})
\le
W(D_t).
\)
The left-hand side is precisely the cost incurred at time $t$ in classical weighted paging, where only pages that remain in cache are loaded. Equality need not hold in \pwtw, since pages may be loaded solely to satisfy requests and evicted immediately afterwards.

For a fixed sequence of cache configurations $C$, we define $\cost_R(C)
=
\min_D \cost_R(C,D)$.

\paragraph{Critical pages.}
For every time $t$, define $E_t := C_t \setminus C_{t-1}$ as the set of pages that enter the cache at time $t$. By the consistency condition, every such page must be loaded at time $t$, and therefore $E_t \subseteq D_t \cap C_t$. We refer to $E_t$ as the set of \emph{critical pages} at time $t$.

\subsection{A Reinterpretation of Gupta et al. \cite{GKP20}}\label{subsection:Reinterpretation}

The separation between cache configurations and page loads allows us to reinterpret the guarantees of Gupta et al.~\cite{GKP20} in terms of the quantity $\cost_R(C)$. Consider any integral solution to their hitting-set LP whose objective value is $M$, and let $C$ be the cache trajectory obtained online from this solution. They provide an online conversion procedure that builds on the cache trajectory and yields a feasible solution. (To this end, recall that it is allowed to load a page $p$ and then immediately evict it, so that $p$ does not appear in the cache configuration at the end of the time step.) They also provide an offline conversion procedure that optimizes the online solution by considering for each page $p$ the times at which such load-and-evict operations occur and removing some of these loads, while preserving feasibility. Since only loads of pages that do not remain in cache are removed, the cache trajectory $C$ remains unchanged.

The resulting offline completion $D$ computed by \cite{GKP20} is feasible and satisfies $\cost_R(C,D)=O(M)$.
Since $\cost_R(C)$ is defined as the minimum cost of any completion of $C$, it follows that $\cost_R(C)\le \cost_R(C,D)=O(M)$.

Thus, every integral LP solution naturally induces a cache trajectory whose completion cost is within a constant factor of the LP objective value.

Gupta et al.~\cite{GKP20} further showed how to construct online an integral solution whose objective value satisfies $M = O(\log k)\cdot \OPT$, where $\OPT$ denotes the cost of an optimal offline solution to \pwtw. We combine the two observations and get the following.

\begin{lemma}\label{lem:good-cache-trajectory}
There exists an online algorithm that constructs a cache trajectory $C$ satisfying $\cost_R(C)=O(\log k)\cdot\OPT$.
\end{lemma}

Lemma~\ref{lem:good-cache-trajectory} shows that an $O(\log k)$-competitive algorithm for \pwtw can be obtained by completing the cache trajectory $C$ produced online by Gupta et al. ~\cite{GKP20}. Thus, it remains to solve the completion problem online: at each time $t$, the algorithm observes $C_t$ and computes a completion $D$, satisfying $\cost_R(C,D)=O(\cost_R(C))$.

\subsection{Joint Replenishment Problem (JRP)}

To formulate the completion problem as a replenishment problem, we introduce a variant of JRP that will be central to our reduction. Joint replenishment is defined over a discrete time horizon $t=1,2,\dots,T$ and a set of items $Y$, where each item $i\in Y$ has cost $w_i$. In our setting, items correspond to pages.

Each request $r$ is specified by a triple \((i(r),s(r),h_r)\),
where $i(r)\in Y$ is the requested item, $s(r)$ is the release time, and $h_r(t)$ is the delay cost incurred if $r$ is served at time $t\ge s(r)$.

A solution consists of choosing, for each time $t$, a set of items $B_t\subseteq Y$ to be served at time $t$. If $B_t\neq\emptyset$, we say that an order is placed at time $t$. A request $r$ is served at the earliest time $t\ge s(r)$ such that $i(r)\in B_t$. The cost of a solution has three components:

\begin{itemize}
    \item \emph{Opening cost.} For each time $t$ with $B_t\neq\emptyset$, we pay a cost $o$ for opening the warehouse.
    \item \emph{Item cost.} For each time $t$ and each item $i\in B_t$, we pay $w_i$.
    \item \emph{Delay cost.} For each request $r$ served at time $t$, we pay $h_r(t)$.
\end{itemize}

The objective is to minimize the total cost.

\paragraph{JRP with deadlines.}
A special case of JRP is obtained when each request $r$ is specified by a triple $(i(r),s(r),e(r))$, where $e(r)$ is its deadline, and must be served within the interval $[s(r),e(r)]$. This corresponds to the delay-cost function
\[
h_r(t)=
\begin{cases}
0 & t\in[s(r),e(r)],\\
+\infty & \text{otherwise}.
\end{cases}
\]
We refer to this variant as JRP with deadlines (JRP-D).

\paragraph{JRP-D with time-dependent opening costs.}\label{definition-variant}

We consider a variant of JRP-D in which the fixed opening cost $o$ is replaced by a time-dependent opening cost $o_t$. This additional flexibility makes the problem substantially harder in the online setting. Consider an instance with a single item $i$ of cost $w_i=0$ and a single request $r=(i,1,2)$. Let the opening cost at time $1$ be $o_1=M$. At time $2$, the opening cost is either $o_2=1$ or $o_2=M^2$, and its value is not known to the online algorithm at time $1$.

If the algorithm waits until time $2$, then when $o_2=M^2$ it pays $M^2$, while the optimum pays $M$ by serving $r$ at time $1$. Thus, any bounded-competitive algorithm must serve $r$ at time $1$. However, when $o_2=1$, the algorithm then pays $M$, while the optimum waits until time $2$ and pays $1$. Hence, the competitive ratio is at least $M$, which is unbounded. Therefore, no online algorithm can achieve a bounded competitive ratio without additional information on future opening costs.

To get around this difficulty, we consider a prediction-augmented setting. For times $t<t'$, define $O_{t,t'}:=\min_{\tau\in(t,t']}\set{o_\tau}$, the minimum opening cost during the interval $(t,t']$. For a request $r$ to item $i$, define $P(r,t):=\max\set{O_{t,e(r)},w_{i}}$. Thus, $P(r,t)$ may overestimate the minimum opening cost $O_{t,e(r)}$ in cases where $w_i>O_{t,e(r)}$. 

The prediction value of a request $r$ at time $t$, denoted by $\widetilde P(r,t)$, is defined as follows. If there exists a request $r'$ whose time window satisfies $[s(r'),e(r')]\subseteq(t,e(r)]$, then $\widetilde P(r,t)$ may be any value in $[0,P(r,t)]$; otherwise, $\widetilde P(r,t)=P(r,t)$.

We next specify which prediction values are available to the algorithm. A request is \emph{active} at time $t$ if it has arrived, but has not yet been served, and an item is \emph{active} if it has an active request. During the construction of an order, requests for items already committed to that order are treated as inactive for subsequent prediction queries. All committed items must be served by the order placed at the same time step.

The algorithm has access to $\widetilde P(r,t)$ only when $r$ has the earliest deadline among the requests active at that point. Thus, committing items to an order may make predictions for other active requests available.

\section{A Reduction for Completing Cache Configurations}\label{sec:completion-reduction}
We study here the completion problem: given a sequence of cache configurations $C$, what is the minimum cost of completing it into a feasible paging solution? We show that this problem is captured, up to a constant factor, by a suitable instance $J(C)$ of the JRP-D variant introduced in Section~\ref{definition-variant}.
We prove that the reduction preserves the completion problem up to a constant factor. First, any completion of $C$ induces a feasible solution to $J(C)$ with comparable cost. Second, any solution to $J(C)$ can be transformed online into a feasible completion of $C$ with no asymptotic loss. Consequently, $J(C)$ characterizes $\cost(C)$ up to a constant factor.

\subsection[]{Defining the Instance $J(C)$}\label{sec:from-paging-to-jrp}

Fix a sequence $R$ of page requests and a cache trajectory $C$. We construct an instance $J_R(C)$ of JRP-D with time-dependent opening costs. Since $R$ is fixed throughout the discussion, we omit it from the notation and write $J(C)$ and $\cost(C)$. To distinguish between the two problems, we refer to requests in $R$ as \emph{page requests} and to requests in $J(C)$ as \emph{JRP requests}.

The instance $J(C)$ has three ingredients: items, JRP requests, and time-dependent opening costs. The items correspond to pages, the JRP requests encode the page loads required for feasibility, and the opening cost at time $t$ captures the additional cost of performing a load while following the cache trajectory $C$.

\paragraph{Items.}
For each page $p\in P$, there is a corresponding item $i_p\in Y$ with item cost $w_{i_p}=w_p$.

\paragraph{Opening costs.}

By the non-degeneracy condition (see~\ref{subsec:our-approach}), every nonempty load $D_t$ must contain a page from $C_t$. Thus, any load performed at time $t$ incurs cost at least
$\min\set{w_p | p\in C_t}$.

Motivated by this observation, we define the opening cost at time $t$ by
\[
o_t := \min\set{ w_p|p \in C_t }.
\]
Observe that this may overestimate the actual cost incurred by a load. For example, if at time $t$ a completion loads a single page $p$ that remains in cache, that is $D_t=\set{p}$, then the paging cost incurred at time $t$ is $w_p$, whereas the corresponding order in $J(C)$ incurs cost
\(
w_p+o_t
=
w_p+\min\set{w_q\mid q\in C_t}
\le 2w_p.
\)
Thus, the reduction can increase the cost by at most a factor of two, which is reflected in our analysis.

\paragraph{Prediction availability.}
Suppose that at time $t<e(r)$, $r$ is the active
request with the earliest deadline, breaking ties
by maximum page weight. Set $\widetilde P(r,t):=w_{p(r)}$. If there exists a request $r'$ with $[s(r'),e(r')]\subseteq(t,e(r)]$, this prediction is valid because our prediction model permits underestimation and $w_{p(r)}\le P(r,t)$.

Otherwise, let $Q$ be the set of remaining active requests with deadline $e(r)$. No remaining active request outside $Q$, and no unrevealed request, has deadline at most $e(r)$. By Lemma~\ref{lem:deadline-cache-repair}, without loss of generality we may assume that some request in $Q$ has its page $q$ in $C_{e(r)}$. By the tie-breaking rule,
$w_q\le w_{p(r)}$, and hence
\[
O_{t,e(r)}
\le o_{e(r)}
\le w_q
\le w_{p(r)}.
\]
Therefore, $\widetilde P(r,t)=w_{p(r)}=P(r,t)$,
as required by the prediction model in this case.

\paragraph{JRP requests.}
The JRP requests in $J(C)$ encode two requirements: serving page requests and performing page loads forced by the cache trajectory. Let
\[
R':=\{r\in R:p(r)\notin C_{s(r)}\}.
\]
These are exactly the page requests that are not served by the cache at their release time and must therefore be satisfied by a subsequent page load.

For every page request $r\in R'$, we introduce a JRP request for item $i_{p(r)}$ with release time $s(r)$ and deadline $e(r)$. In addition, for every time $t$ and every page $p\in E_t$, we introduce a JRP request for item $i_p$ with release time and deadline both equal to $t$.

The first type of JRP request ensures that every page request in $R'$ is served within its time window. The second type enforces consistency with the prescribed cache trajectory: if page $p$ enters the cache at time $t$, then every feasible completion must load $p$ at time $t$.

Finally, the reduction must preserve the online nature of the problem: at time $t$, the algorithm must be able to determine $o_t$ and all JRP requests arriving at time $t$ using only information available up to time $t$.

\begin{theorem}\label{thm:online-instance}
Given a cache trajectory $C$ revealed online, the instance $J(C)$ can be constructed online. In particular, at each time $t$, the opening cost $o_t$ and all JRP requests arriving at time $t$ depend only on the history up to time $t$.
\end{theorem}

\subsection[]{Upper Bounding $\OPT(J(C))$}\label{sec:paging-bound-jrp}

We show that any feasible paging solution induces a feasible solution to $J(C)$ with comparable cost. In particular, given any sequence of page loads $D$ that completes $C$, we can construct a solution to $J(C)$ whose cost is within a constant factor of $\cost_R(C,D)$. As a consequence, $\OPT(J(C))\le O(\cost_R(C))$.

Fix any loading sequence $D$ such that $(C,D)$ is feasible for $R$. Define the associated JRP solution $S^\star=(S^\star_t)_{t=1}^T$ by
\[
S^\star_t:=\{i_p:p\in D_t\}
\qquad\text{for every }t.
\]
Thus, the order $S^\star_t$ serves precisely the items corresponding to the pages loaded at time $t$.

\paragraph{Feasibility.}
We claim that $S^\star$ is feasible for $J(C)$. Consider a page request $r\in R'$ for page $p=p(r)$. Since $(C,D)$ is feasible for $R$, there exists a time $t\in[s(r),e(r)]$ such that $p\in D_t$. Hence $i_p\in S^\star_t$, and the corresponding JRP request is served within its time window.

For every time $t$ and page $p\in E_t=C_t\setminus C_{t-1}$, we have $p\in D_t$, and thus $i_p\in S^\star_t$. Hence, the corresponding JRP request, whose release time and deadline are both $t$, is also served. Thus, $S^\star$ is feasible for $J(C)$.

\paragraph{Cost analysis.}
The item-cost part satisfies
\[
\sum_t\sum_{i_p\in S^\star_t}w_{i_p}
=
\sum_t\sum_{p\in D_t}w_p
=
\cost_R(C,D).
\]
It remains to bound the opening-cost part. Fix a time $t$ such that $D_t\neq\emptyset$. By the non-degeneracy condition, $D_t\cap C_t\neq\emptyset$. Therefore,
\[
o_t
=
\min_{p\in C_t}w_p
\le
\min_{p\in C_t\cap D_t}w_p
\le
\sum_{p\in D_t}w_p.
\]
Summing over all times $t$ with $D_t\neq\emptyset$, we obtain
\[
\sum_{t:\,D_t\neq\emptyset}o_t
\le
\sum_t\sum_{p\in D_t}w_p
=
\cost_R(C,D).
\]

\paragraph{Conclusion.}
We conclude that for every feasible $(C,D)$,
\[
\OPT(J(C))
\le
\cost(S^\star)
\le
2\cdot\cost_R(C,D).
\]
Taking the minimum over all feasible loading sequences $D$, we obtain
\[
\OPT(J(C))
\le
2\cdot\cost_R(C).
\]

\subsection{From JRP-D Solutions to Paging Solutions}\label{sec:from-jrp-to-paging}

We now show how to complete a cache trajectory $C$ into a feasible paging solution using a solution to the instance $J(C)$.

Let $S=(S_t)_{t=1}^T$ be a feasible solution to $J(C)$, where $S_t\subseteq Y$ is the order placed at time $t$. Our goal is to construct a sequence of page loads $D$ such that $(C,D)$ is feasible and has cost comparable to that of $S$.

\paragraph{Construction.}
We define the loading sequence $D$ online as follows. For every time $t$, let $q_t$ be a minimum-weight page in $C_t$, i.e., $q_t\in\arg\min_{p\in C_t}w_p$, and define
\[
D_t :=
\begin{cases}
\emptyset, & \text{if } S_t=\emptyset,\\
\{p:i_p\in S_t\}\cup\{q_t\}, & \text{otherwise}.
\end{cases}
\]

Intuitively, when $S_t\neq\emptyset$, we load the pages corresponding to the items in $S_t$. If $C_t=C_{t-1}$, we temporarily evict $q_t$ to create space in the cache, and then load and immediately evict  all other pages corresponding to items in $S_t$ (one by one). We then reload $q_t$ to restore the cache configuration $C_t$. When $C_t\neq C_{t-1}$, evicting and reloading $q_t$ may be unnecessary. For simplicity, we include $q_t$ in $D_t$ and charge its cost to the opening cost $o_t$ of the corresponding JRP order.

\paragraph{Feasibility.}
We show that $(C,D)$ is a feasible paging solution.

\emph{Capacity.} The capacity constraint depends only on $C$ and is therefore satisfied.

\emph{Service.} Let $r$ be a page request. If $p(r)\in C_{s(r)}$, then $r$ is already served. Otherwise, $J(C)$ contains a JRP request for item $i_{p(r)}$ with release time $s(r)$ and deadline $e(r)$. Since $S$ is feasible, there exists a time $t\in[s(r),e(r)]$ such that $i_{p(r)}\in S_t$. Hence $p(r)\in D_t$, and $r$ is served.

\emph{Consistency.} Fix a time $t$. If $C_t=C_{t-1}$, the consistency condition is immediate. Otherwise, for every page $p\in E_t$, the instance $J(C)$ contains a JRP request for item $i_p$ with release time and deadline both equal to $t$. Since $S$ is feasible, $i_p\in S_t$, and hence $p\in D_t$. Thus, $C_t\setminus C_{t-1}\subseteq D_t$.

\emph{Non-degeneracy.} Fix a time $t$ such that $D_t\neq\emptyset$. By construction, $q_t\in D_t\cap C_t$, and therefore $D_t\cap C_t\neq\emptyset$.

\paragraph{Cost analysis.}
If $S_t=\emptyset$, then $D_t=\emptyset$ and no cost is incurred. Otherwise,
\[
\sum_{p\in D_t}w_p
\le
\sum_{i_p\in S_t}w_{i_p}+w_{q_t}.
\]
By the definition of $q_t$,
\[
w_{q_t}=\min_{p\in C_t}w_p=o_t.
\]
Therefore,
\[
\sum_{p\in D_t}w_p
\le
\sum_{i_p\in S_t}w_{i_p}+o_t.
\]

Summing over all times $t$ such that $S_t\neq\emptyset$, we obtain
\[
\cost_R(C,D)
\le
\sum_t\sum_{i_p\in S_t}w_{i_p}
+
\sum_{t:\,S_t\neq\emptyset}o_t
=
\cost(S).
\]
Thus, \(\cost_R(C,D)\le\cost(S)\), which completes the proof.

\subsection{Consequences for Paging with Time Windows}\label{sec:combining-results}

We now show that a constant-competitive algorithm for JRP-D with time-dependent opening costs, as introduced in Section~\ref{definition-variant}, immediately implies an $O(\log k)$-competitive algorithm for \pwtw.

\begin{theorem}\label{thm:online-completion}
Assume there exists a constant-competitive online algorithm for JRP-D with time-dependent opening costs under the prediction model of Section~\ref{definition-variant}. Then there exists an online algorithm that, given a sequence $R$ of page requests and a cache trajectory $C$, computes a completion sequence $D$ online such that $(C,D)$ is feasible for $R$ and
\[
\cost_R(C,D)=O(\cost_R(C)).
\]
\end{theorem}

\begin{proof}
Given $R$ and $C$, we construct the instance $J_R(C)$ online as in Section~\ref{sec:from-paging-to-jrp}. We then run the assumed constant-competitive algorithm on $J_R(C)$, obtaining an online solution $S$ satisfying
\[
\cost(S)=O(\OPT(J_R(C))).
\]
Using the reduction of Section~\ref{sec:from-jrp-to-paging}, we convert $S$ online into a feasible paging solution $(C,D)$ such that
\[
\cost_R(C,D)\le\cost(S).
\]
Finally, by Section~\ref{sec:paging-bound-jrp},
\[
\OPT(J_R(C))\le O(\cost_R(C)).
\]
Combining these inequalities yields
\[
\cost_R(C,D)=O(\cost_R(C)).
\]
\end{proof}

Combining Theorem~\ref{thm:online-completion} with Lemma~\ref{lem:good-cache-trajectory}, we obtain the following.

\begin{theorem}\label{thm:main-conditional}
Assume there exists a constant-competitive online algorithm for JRP-D with time-dependent opening costs under the prediction model of Section~\ref{definition-variant}. Then there exists an $O(\log k)$-competitive online algorithm for paging with time windows.
\end{theorem}

\begin{proof}
Let $C$ be the cache trajectory produced by the algorithm of Lemma~\ref{lem:good-cache-trajectory}. By Theorem~\ref{thm:online-completion}, we can compute online a loading sequence $D$ satisfying

\[
\cost_R(C,D)=O(\cost_R(C)).
\]

By Lemma~\ref{lem:good-cache-trajectory},

\[
\cost_R(C)=O(\log k)\cdot\OPT.
\]

Combining the bounds yields

\[
\cost_R(C,D)=O(\log k)\cdot\OPT.
\]
\end{proof}

The next section is devoted to proving the assumption in Theorems~\ref{thm:online-completion} and~\ref{thm:main-conditional} on the existence of a constant-competitive online algorithm for JRP-D with time-dependent opening costs.

\section{JRP-D with Time-Dependent Opening Costs}\label{sec:jrp-algorithm}

In this section, we study the variant of JRP-D with time-dependent opening costs introduced in Section~\ref{definition-variant}. Recall that placing an order at time $t$ incurs opening cost $o_t$, and that the algorithm has access to the predictions described in Section~\ref{definition-variant}. Throughout this section, we assume that $o_t$ is a power of two for every $t$. This can be achieved by rounding each opening cost up to the next power of two, increasing the cost of any solution by at most a factor of two.

We now present a constant-competitive online algorithm that combines the classical batching approach of Bienkowski \etal~\cite{bienkowski2014better} with predictions of future opening costs to handle their variability over time. 

\subsection{Algorithm}
For an active item $i$, we refer to the deadline
of its earliest-deadline active request as the
\emph{active deadline} of $i$. At each time $t$,
the algorithm constructs a set $X_t$ of items
and places an order for $X_t$ if it is nonempty.

The algorithm consists of two stages.
Stage~I initiates an order either to meet a
deadline or based on a prediction of future
opening costs. It then greedily adds active
items as long as the total cost of $X_t$ does
not exceed the opening cost $o_t$.
Stage~II uses predictions to determine which
additional active items to include.
Throughout both stages, the algorithm maintains
and processes the active items in nondecreasing
order of their active deadlines, breaking ties
by nonincreasing item cost and then by item index.

\vspace{0.1in}

\noindent
The algorithm performs the following at each time $t$:
\paragraph{Stage I: Deadline items and greedy batching.}

The algorithm initializes
\[
X_t=\{i:\text{the active deadline of $i$ is $t$}\}.
\]

\noindent
If $X_t=\emptyset$, let $r$ be the first active request in the maintained order. If no active request exists, or if
\[
\widetilde P(r,t)\le o_t,
\]
the algorithm proceeds directly to time $t+1$.

Otherwise, if $W(X_t)\le o_t$, the algorithm considers the remaining active items and greedily adds them to $X_t$ as long as the total weight does not exceed $o_t$. The greedy procedure stops when the next item cannot be added without exceeding $o_t$. 

The algorithm then proceeds to Stage~II.
\paragraph{Stage II: Prediction-based orders.}
The algorithm continues with the set $X_t$ constructed in Stage~I and considers the remaining active items in the maintained order. For each such item $i$, let $r$ be its earliest-deadline active request. The algorithm stops if either of the following conditions holds:
\begin{enumerate}
    \item \(\widetilde P(r,t)\le o_t\);
    \item \(\widetilde P(r,t)\le W(X_t)+w_i\).
\end{enumerate}
Otherwise, the algorithm adds $i$ to $X_t$ and
continues to the next item.

After stopping, or after considering all remaining
active items, the algorithm places an order for
$X_t$ if it is nonempty and proceeds to time $t+1$.

\paragraph{Prediction availability.}
During the construction of $X_t$, requests for items already included in $X_t$ are treated as inactive. Whenever the algorithm queries $\widetilde P(r,t)$, the request $r$ has the earliest deadline among the remaining active requests, and its item has maximum cost among the remaining active items with that deadline. Thus, every prediction queried by the algorithm is available under our prediction model.

\subsection{Competitive Analysis}

We prove that the algorithm is constant-competitive. Fix an optimal offline solution, denoted by $\OPT$, and partition the time horizon into phases according to the order times of $\OPT$.

\begin{definition}
A \emph{phase} is a half-open interval $[g,h)$ between two consecutive order times of $\OPT$.
\end{definition}

Fix a phase $[g,h)$, and let
$t_1<t_2<\cdots<t_\ell$ be the order times of the
algorithm during this phase. For convenience,
define $t_0=g$. We represent each item service
by a pair $(i,t)$, indicating that item $i$ is
served at time $t$. We partition these pairs
into two sets:
\begin{itemize}
    \item $M$: pairs $(i,t)$ for which the active
    deadline of $i$ at $t$ is at most $h-1$.
    \item $N$: pairs $(i,t)$ for which the active
    deadline of $i$ at $t$ is at least $h$.
\end{itemize}
When $t$ is clear from context, we write $i\in M$ instead of $(i,t)\in M$, and similarly for $N$. We also refer to items from $M$ or $N$. Note that the same item may appear in both $M$ and $N$ during a phase, through services at different times.

We analyze the costs associated with $M$ and $N$ separately. We show that the costs associated with $M$ can be charged to the corresponding item costs incurred by $\OPT$ and to its opening cost at time $g$. The costs associated with $N$ are bounded in terms of the opening cost incurred by $\OPT$ at time $h$.

The interval preceding the first $\OPT$ order is treated as an initial phase containing only items from $N$, whose costs are charged to the opening cost of the first $\OPT$ order. Similarly, the interval following the last $\OPT$ order is treated as a final phase containing only items from $M$. Its costs are bounded by the analysis of $M$ and charged to the opening cost of the last $\OPT$ order and the corresponding item costs incurred by $\OPT$.

Each item cost and opening cost incurred by $\OPT$ is charged at most a constant number of times. Thus, summing these bounds over all phases yields the following theorem.

\begin{theorem}\label{thm:jrp-2}
The algorithm is $O(1)$-competitive.
\end{theorem}

\subsubsection{Bounding the Cost of $M$}

The total item cost incurred by the algorithm for services in $M$ can be charged to the corresponding item costs incurred by $\OPT$. Indeed, since $\OPT$ places no orders during $(g,h)$, the requests defining the active deadlines of these items must be served by $\OPT$ by time $g$. 

\begin{lemma}\label{lem:M-item-costs}
The total item cost incurred by the algorithm for
items in $M$, summed over all phases, is at most
the total item cost incurred by $\OPT$.
\end{lemma}

\begin{proof}
Fix a phase $[g,h)$. Since $\OPT$ places no orders
in $(g,h)$, every request defining the active
deadline of an item in $M$ must be served by
$\OPT$ no later than time $g$.

Moreover, each item appears in $M$ at most once
during the phase. Indeed, whenever the algorithm
serves an item, it satisfies all its currently
active requests. Any subsequent appearance of
that item must therefore be due to a request
arriving after its previous service, which cannot
have a deadline before $h$.

Finally, between two consecutive appearances of
the same item in $M$ across different phases,
$\OPT$ must also serve that item. Hence, each
appearance can be charged to a distinct item
cost incurred by $\OPT$.
\end{proof}

By Lemma~\ref{lem:M-item-costs}, it remains to bound the opening costs of orders serving items from $M$. We will use the following observation throughout the analysis.

\begin{lemma}\label{lem:exact-prediction}
Let $r$ be an active request at time $t\in[g,h)$ with $t<e(r)<h$. Then, $\widetilde P(r,t)=P(r,t)$.
\end{lemma}

\begin{proof}
Suppose there exists a request $r'$ such that \([s(r'),e(r')]\subseteq(t,e(r)]\). Then every feasible solution must place an order serving $r'$ during $(t,e(r)]\subseteq(g,h)$, contradicting the fact that $\OPT$ places no orders in $(g,h)$. Hence, by the definition of the prediction model, $\widetilde P(r,t)=P(r,t)$. 
\end{proof}

Let $j$ be the last index such that the order placed at time $t_j$ contains an item from $M$. We bound the opening costs of the orders serving items from $M$ using two lemmas: the first bounds the opening costs of all orders preceding $t_j$, while the second bounds the opening cost at $t_j$.

\begin{lemma}\label{lem:M-intermediate}
For every $1\le\ell<j$, let $i_{\ell+1}$ denote the first item from $M$ served by the order placed at time $t_{\ell+1}$. Then,
\[
o_{t_\ell}<W(X_{t_\ell})+w_{i_{\ell+1}}.
\]
\end{lemma}
\begin{proof}
Fix $1\le \ell<j$. Since $i_{\ell+1}\in M$, it is already active at time $t_\ell$ and is the first remaining item in deadline order after the order at $t_\ell$. If Stage~II adds any item, it must first add the item that terminated the greedy procedure in Stage~I, implying $W(X_{t_\ell})>o_{t_\ell}$.

Otherwise, Stage~II adds no items and $i_{\ell+1}$ is the item that terminated the greedy procedure in Stage~I, and hence
\[
o_{t_\ell}
<
W(X_{t_\ell})+w_{i_{\ell+1}}.
\]
Thus, the claimed inequality holds in all cases.
\end{proof}
It remains to bound the opening cost of the last order serving an item from $M$, placed at time $t_j$. If $t_j=g$, its opening cost can be charged directly to $\OPT$. Otherwise, recall that $t_0=g$. If $o_{t_j}\le o_{t_{j-1}}$, its opening cost can be charged to the preceding opening cost (or to $\OPT$ when $j=1$). The following lemma handles the remaining case.

\begin{lemma}\label{lem:M-last}
Suppose $t_j>g$ and $o_{t_j}>o_{t_{j-1}}$.
Let $i_j$ be the first item from $M$ served by
the order placed at time $t_j$. Then,
\[
o_{t_j}\le W(X_{t_{j-1}})+w_{i_j}.
\]
\end{lemma}

\begin{proof}
If $o_{t_j}\le w_{i_j}$, the claim is immediate. Hence, assume that $o_{t_j}>w_{i_j}$. Let $r$ be the earliest-deadline active request of $i_j$ at time $t_{j-1}$. Since $i_j\in M$, we have $e(r)<h$, and therefore, by Lemma~\ref{lem:exact-prediction}, $\widetilde P(r,t)=P(r,t)$ for every $t\in[t_{j-1},t_j]$ satisfying $t<e(r)$.

We claim that $o_{t_j}$ is the minimum opening cost in $(t_{j-1},e(r)]$. First, suppose there exists $t'\in(t_j,e(r)]$ such that $o_{t'}<o_{t_j}$. Then
\[
\widetilde P(r,t_j)
=
P(r,t_j)
=
\max\{O_{t_j,e(r)},w_{i_j}\}
\leq\max\{o_{t'},w_{i_j}\}
<
o_{t_j}.
\]
Since $e(r)>t_j$, the item $i_j$ is not due at time $t_j$. Moreover, since $i_j$ is the first item served at $t_j$, no active item has an earlier deadline. Thus, Stage~I initializes $X_{t_j}=\emptyset$.

As the earliest active request satisfies
$\widetilde P(r,t_j)<o_{t_j}$,
Stage~I proceeds directly to time $t_j+1$
without placing an order, contradicting
the fact that $i_j$ is served at $t_j$.

Next, suppose there exists $t'\in(t_{j-1},t_j)$ such that $o_{t'}<o_{t_j}$, and let $t'$ be the last such time. By the preceding argument, $o_{t_j}$ is the minimum opening cost in $[t_j,e(r)]$. Consequently,
\[
\widetilde P(r,t')
=
P(r,t')
=
\max\{O_{t',e(r)},w_{i_j}\}
=\max\{o_{t_j},w_{i_j}\}
=
o_{t_j}.
\]

Since no order is placed by the algorithm
between $t_{j-1}$ and $t_j$, Stage~I
initializes $X_{t'}=\emptyset$.
Moreover, $i_j$ is the first active item
in deadline order at time $t'$.

Since
\[
\widetilde P(r,t')=o_{t_j}>o_{t'},
\]
Stage~I performs greedy batching.

If $w_{i_j}\le o_{t'}$, then Stage~I
adds $i_j$ to $X_{t'}$.

Otherwise, $w_{i_j}>o_{t'}$, so Stage~I
cannot add $i_j$. However, \(
\widetilde P(r,t')
=
o_{t_j}
>
\max\{o_{t'},w_{i_j}\},
\) and hence neither stopping condition of
Stage~II holds when $i_j$ is considered.
Thus, Stage~II adds $i_j$ to $X_{t'}$.

In either case, the algorithm serves $i_j$ at time $t'<t_j$, a contradiction. Therefore, $o_{t_j}$ is the minimum opening cost in $(t_{j-1},e(r)]$, proving the claim. In particular,
\[
\widetilde P(r,t_{j-1})
=
P(r,t_{j-1})
=
\max\{o_{t_j},w_{i_j}\}
=
o_{t_j}.
\]
If $j=1$ and $t_1>g$, then the algorithm places no order at $g$. Since $\widetilde P(r,g)=o_{t_1}>o_g$, Stage~I cannot skip time $g$. If $w_{i_1}\le o_g$, greedy batching adds $i_1$. Otherwise, since $o_{t_1}>\max\{o_g,w_{i_1}\}$, Stage~II adds $i_1$. In either case, the algorithm serves $i_1$ at $g$, a contradiction. Hence, we may assume $j\ge 2$.

Since the request of $i_j$ remains active after the order at $t_{j-1}$, $i_j$ was not added during Stage~I. Moreover, it is the earliest active item remaining after that order. Therefore, Stage~II stops when considering $i_j$.

The first stopping condition cannot hold, since
\[
o_{t_{j-1}}
<
o_{t_j}
=
\widetilde P(r,t_{j-1}).
\]
Therefore, the second stopping condition
must hold, giving
\[
W(X_{t_{j-1}})+w_{i_j}
\ge
\widetilde P(r,t_{j-1})
=
o_{t_j},
\]
as required.
\end{proof}

Combining Lemmas~\ref{lem:M-intermediate}
and~\ref{lem:M-last}, we obtain
\[
\sum_{\ell=1}^{j} o_{t_\ell}
=
O\left(
o_g+
\sum_{(i,t)\in M} w_i
\right).
\]
Indeed, every order preceding $t_j$ contains only
items from $M$. Applying Lemma~\ref{lem:M-intermediate} to each
order preceding $t_j$ bounds its opening cost by the total item cost of that order and the cost of the first item served in the next order.
Since all these services belong to $M$, each
item service is charged at most twice. The final opening cost is bounded by $o_g$, the preceding opening cost, or by Lemma~\ref{lem:M-last}. By Lemma~\ref{lem:M-item-costs}, these costs,
summed over all phases, are bounded by a constant
multiple of the cost incurred by $\OPT$.

\subsubsection{Bounding the Cost of $N$}

We next consider items in $N$. We show that the opening costs of the relevant orders are strictly increasing. Since the opening costs are powers of two, this allows us to bound their total cost by $O(o_h)$ and charge it to the order placed by $\OPT$ at time $h$.

Recall that $t_j$ is the last order time during the phase at which the algorithm serves an item from $M$. Observe that every order preceding $t_j$ contains only items from $M$. Indeed, all requests with deadlines before $h$ must have arrived by time $g$, since $\OPT$ places no orders in $(g,h)$. Moreover, both stages process items in nondecreasing order of their active deadlines and cannot skip an item to serve one with a later deadline. Thus, an order containing an item from $N$ must serve all remaining items from $M$, and hence cannot precede $t_j$.

The total cost of items from $N$ added during Stage~I at $t_j$ is at most $o_{t_j}$, and is therefore accounted for by the bound on the opening cost $o_{t_j}$ in the analysis of $M$. We therefore consider the order at $t_j$ only if Stage~II adds an item from $N$, together with all subsequent orders in the phase.

Let
\[
e_1<e_2<\cdots<e_q
\]
denote these order times. If $M=\emptyset$, we instead consider all orders in the phase. For every such order at time $e_\ell$, one of the following holds:
\begin{enumerate}
    \item Stage~II adds at least one item from $N$.
    \item Stage~II adds no items, and the order
    consists entirely of items added during
    Stage~I.
\end{enumerate}

Indeed, after $t_j$, no active request has
deadline before $h$, so Stage~I initializes
$X_t=\emptyset$.

\begin{lemma}\label{lem:N-opening-increasing}
For every $\ell<q$,
\[
o_{e_\ell}<o_{e_{\ell+1}}.
\]
Moreover, for every $\ell$,
\[
o_{e_\ell}<o_h.
\]
\end{lemma}

\begin{proof}
Fix $\ell$ and consider the order at time
$e_\ell$.

First, suppose Stage~II adds an item $i\in N$,
and let $r$ be its earliest-deadline active
request when it is considered.

Since neither stopping condition holds,
\[
\widetilde P(r,e_\ell)>o_{e_\ell}
\]
and
\[
\widetilde P(r,e_\ell)>
W(X_{e_\ell}^{\mathrm{before}})+w_i
\ge w_i,
\]
where $X_{e_\ell}^{\mathrm{before}}$ denotes
the set immediately before adding $i$.

Using $\widetilde P(r,e_\ell)\le P(r,e_\ell)$,
we obtain
\[
o_{e_\ell}
<
\widetilde P(r,e_\ell)
\le
\max\{O_{e_\ell,e(r)},w_i\}.
\]
Since $w_i<\widetilde P(r,e_\ell)$, it follows that $o_{e_\ell}<O_{e_\ell,e(r)}$.

Now suppose Stage~II adds no items. Let $i$ be the first item added during Stage~I, and let $r$ be its earliest-deadline active request. Since $X_{e_\ell}$ was initially empty, $r$ is the earliest active request at time $e_\ell$. Thus, by the condition for initiating greedy batching, $\widetilde P(r,e_\ell)>o_{e_\ell}$.

Moreover, Stage~I adds $i$ while maintaining $W(X_{e_\ell})\le o_{e_\ell}$, so $w_i\le o_{e_\ell}$.
Consequently,
\[
o_{e_\ell}
<
\widetilde P(r,e_\ell)
\le
\max\{O_{e_\ell,e(r)},w_i\}
\]
implies $o_{e_\ell}<O_{e_\ell,e(r)}$. Thus, in either case, there exists an item $i\in N$ with active request $r$ such that \(o_{e_\ell}<O_{e_\ell,e(r)}\). Since $i\in N$, its active deadline satisfies $e(r)\ge h$. Therefore, for every $\ell<q$,
\[
o_{e_\ell}<o_{e_{\ell+1}}.
\]
The same argument, applied to $h$ instead of $e_{\ell+1}$, gives $o_{e_\ell}<o_h$ for every $\ell$.
\end{proof}

It remains to bound the item costs of these
orders.

Fix $\ell$. If Stage~II adds an item from $N$,
let $i_\ell$ be the last such item, and let
$r_\ell$ be its earliest-deadline active request
when it is considered. Since the second stopping
condition does not hold,
$W(X_{e_\ell})<\widetilde P(r_\ell,e_\ell)$.
Moreover,
\[
w_{i_\ell}\le W(X_{e_\ell})
<\widetilde P(r_\ell,e_\ell)
\le\max\{O_{e_\ell,e(r_\ell)},w_{i_\ell}\}.
\]
It follows that
$W(X_{e_\ell})<O_{e_\ell,e(r_\ell)}$.
Since $e(r_\ell)\ge h$, we obtain
$W(X_{e_\ell})<o_{e_{\ell+1}}$ for $\ell<q$,
and $W(X_{e_q})<o_h$.

Otherwise, Stage~II adds no items, so all
items in the order were added during Stage~I,
which guarantees $W(X_{e_\ell})\le o_{e_\ell}$.
By Lemma~\ref{lem:N-opening-increasing},
the same bounds follow.

Thus, in either case,
\[
\sum_{\ell=1}^q W(X_{e_\ell})
\le \sum_{\ell=1}^{q-1}o_{e_{\ell+1}}+o_h.
\]

By Lemma~\ref{lem:N-opening-increasing},
the opening costs $o_{e_1},\ldots,o_{e_q}$
form a strictly increasing sequence of
powers of two, with $o_{e_q}<o_h$.
Consequently, both the total opening cost
and the total item cost are $O(o_h)$, giving
\[
\sum_{\ell=1}^q
\bigl(o_{e_\ell}+W(X_{e_\ell})\bigr)
=O(o_h).
\]

Since $\OPT$ places an order at time $h$,
these costs can be charged to its opening
cost $o_h$.

\section*{AI Usage Disclosure}
The authors used ChatGPT to refine and polish the
text. All suggestions were reviewed and edited
by the authors, who take full responsibility
for the final manuscript.

\section*{Submission History}
An earlier version of this paper was submitted
to a conference in July 2026.

\bibliography{sources}

\appendix

\section{Extension to Paging with Delays}\label{app:delays}

We explain how our framework extends to the paging-with-delays model of Gupta et al.~\cite{GKP20}. Gupta et al. show that it suffices to consider a penalty version of paging with time windows. In this variant, each request has a time window and may either be served within its window or left unserved by paying a request-specific penalty. They further show that their analysis for paging with time windows extends to this penalty model.

We follow the same approach. In paging with time windows and penalties, each request $r$ is specified by a page $p(r)$, a release time $s(r)$, a deadline $e(r)$, and a penalty $\pi_r$. A solution may either serve $r$ during $[s(r),e(r)]$ or pay the penalty $\pi_r$. The objective is to minimize the sum of the costs and the penalties of all unserved requests.

Fix an integral solution to the hitting-set LP of Gupta et al.~\cite{GKP20} with objective value $X$. Recall that their conversion procedure (see Section~\ref{subsection:Reinterpretation}) transforms an integral LP solution into a \pwtw solution. The solution consists of two components: the cache trajectory $C=(C_t)$ and the sequence of load-and-evict sets $D=(D_t)$. We now use both components to partition online the requests into two sets, $R_1$ and $R_2$. If a request $r$ is not served (a penalty is paid), we add $r$ to $R_1$. Otherwise, let $t$ be the first time at which the online conversion procedure serves $r$, either through cache residency or by loading $p(r)$, and add the page request $(p(r),t,e(r))$ to $R_2$. Since the partition is determined by the online conversion procedure, both $R_1$ and $R_2$ can be constructed online.

We first bound the completion cost of $C$ with respect to $R_2$. Gupta \etal showed that the offline conversion procedure preserves the cache trajectory $C$ and only removes some of the ``load-and-evict" steps (determined online). Nevertheless, it remains feasible for every request served by the online procedure and therefore serves every request in $R_2$. Since its total cost is $O(X)$, it follows that
\[
\cost_{R_2}(C)=O(X).
\]

Moreover, every request in $R_1$ is left unserved by the online conversion procedure. Since the offline procedure only removes loads, these requests remain unserved, and the offline procedure pays their penalties. Therefore,
\[
\sum_{r\in R_1}\pi_r=O(X).
\]

Since $R_2$ can be constructed online, we apply our online completion algorithm to $C$ and $R_2$, obtaining a loading sequence $D$ satisfying
\[
\cost_{R_2}(C,D)
=
O(\cost_{R_2}(C))
=
O(X).
\]
Now, the penalties paid for the requests in $R_1$ add up to at most $O(X)$. Thus, the resulting solution has total cost $O(X)$. Combining this with the bound on the integral solution of Gupta et al.~\cite{GKP20} yields an $O(\log k)$-competitive online algorithm for paging with delays.

\section{A Local Repair of Cache Trajectories}
\label{app:deadline-cache}

\begin{lemma}\label{lem:deadline-cache-repair}
Let $S$ be a feasible paging solution, and let
$Q$ be a nonempty set of requests active at time
$t_1$, all with deadline $t_2>t_1$.
Suppose that no active or unrevealed request
outside $Q$ has deadline at most $t_2$.
Then $S$ can be modified, without increasing
its cost, so that $p(r)\in C_{t_2}$ for some
$r\in Q$. The modification can be implemented
online whenever $S$ is generated online.
\end{lemma}

\begin{proof}
Follow $S$ until all requests in $Q$ have been
served. Let $r$ be the last request in $Q$ to
be served, and let $K$ be the cache configuration immediately after this
service. In particular, this configuration
contains $p(r)$.

At this point, all requests with deadline at
most $t_2$ have been served. Indeed, every
request in $Q$ has been served, while every
request outside $Q$ with such a deadline
was already served at time $t_1$, by assumption. Postpone all subsequent loadings and evictions through the end of time $t_2$ to the beginning of time $t_2+1$. Execute these operations in their original order, before the operations originally scheduled at that time.

This preserves feasibility: any request whose
service is postponed and that has not already
been served has deadline at least $t_2+1$.
Capacity is preserved because the postponed
operations start from the same
configuration $K$ and are executed in the same order.
No additional loads are introduced, so the cost
does not increase. If $t_2$ is the final time
step, the postponed operations can instead
be omitted. The configuration $K$ persists through
time $t_2$, so $p(r)\in C_{t_2}$.

To implement the modification online, buffer
the postponed operations. If a new request
with deadline at most $t_2$ arrives, execute
the buffered operations before processing it
and resume $S$. This preserves feasibility
even when the hypothesis of the lemma does
not hold. Under the hypothesis, no such
request arrives, so the configuration $K$
persists until the end of time $t_2$, as required.
\end{proof}

Thus, any feasible paging solution can be
modified at no additional cost to satisfy
the stated property for the specified
times $t_1,t_2$ and set $Q$.
\end{document}